\documentclass{elsarticle}

\usepackage{lineno,graphicx,xcolor}
\usepackage{pifont}
\usepackage{verbatim}
\usepackage{amsmath,amsfonts}
\usepackage{enumerate,url}
\usepackage{paralist}
\usepackage{subfig}
\usepackage{wrapfig,array,multirow}
\usepackage{algorithm}
\usepackage[noend]{algpseudocode}

\journal{}

\newtheorem{theorem}{Theorem}
\newtheorem{lemma}[theorem]{Lemma}
\newdefinition{rmk}{Remark}
\newdefinition{definition}{Definition}
\newdefinition{corollary}{Corollary}
\newdefinition{observation}{Observation}
\newproof{proof}{Proof}

\newcommand{\np}{\mathsf{NP}}

\newcommand{\opt}{\mathsf{OPT}}
\newcommand{\opts}{\mathsf{S^{OPT}}}
\newcommand{\el}{\mathsf{L}}
\newcommand{\mis}{\mathsf{MIS}}
\newcommand{\ul}{\ulcorner}
\newcommand{\ur}{\urcorner}
\newcommand{\llc}{\llcorner}
\newcommand{\lr}{\lrcorner}

\newcommand{\rev}[1]{{\color{black}#1}}

\makeatletter
    \renewcommand*{\@fnsymbol}[1]{\ensuremath{\ifcase#1\or *\or \dagger\or \S\or
       \mathsection\or \mathparagraph\or \|\or **\or \dagger\dagger
       \or \ddagger\ddagger \else\@ctrerr\fi}}
\makeatother

\begin{document}

\begin{frontmatter}

\title{Computing Maximum Independent Set on  Outerstring Graphs and Their Relatives}
\tnotetext[mytitlenote]{A preliminary version of this paper appeared in the 16th International Symposium on Algorithms and Data Structures (WADS 2019)~\cite{BoseCKMMMS19}. Research of Prosenjit Bose, Anil Maheshwari, Debajyoti Mondal and Michiel Smid is supported in part by NSERC. Part of this work was done when Saeed Mehrabi was visiting the University of Saskatchewan.}

\author[add1]{Prosenjit Bose}
\ead{jit@scs.carleton.ca}
\author[add2]{Paz Carmi}
\ead{carmip@cs.bgu.ac.il}
\author[add3]{J. Mark Keil}
\ead{keil@cs.usask.ca}
\author[add1]{Anil Maheshwari}
\ead{anil@scs.carleton.ca}
\author[add1]{Saeed Mehrabi}
\ead{}
\author[add3]{Debajyoti Mondal}
\ead{d.mondal@usask.ca}
\author[add1]{Michiel Smid}
\ead{michiel@scs.carleton.ca}

\address[add1]{School of Computer Science, Carleton University, Ottawa, Canada.}
\address[add2]{Department of Computer Science, Ben-Gurion University of the Negev, Beer-Sheva, Israel.}
\address[add3]{Department of Computer Science, University of Saskatchewan, Saskatoon, Canada.}

\begin{abstract}
A graph $G$ with $n$ vertices is called an \emph{outerstring graph} if it has an intersection representation with a set of $n$ curves inside a disk such that one endpoint of every curve is attached to the boundary of the disk. Given an outerstring graph representation of $G$ with $s$ segments, a \emph{Maximum Independent Set} ($\mis$) of $G$ can be \rev{computed} in $O(s^3)$ time (Keil et al., Comput. Geom., 60:19--25, 2017). 

We examine the fine-grained complexity of the $\mis$ problem on some well-known outerstring representations (e.g., line segments, $\el$-shapes, etc.), where the strings are of constant size. We show that computing $\mis$ on grounded segment and grounded square-$\el$ representations is at least as hard as computing $\mis$ on circle graph representations. Note that no $O(n^{2-\delta})$-time algorithm, $\delta>0$, is known for computing $\mis$ on circle graphs. For the grounded string representations, where the strings are $y$-monotone simple polygonal paths of constant length with segments at  integral coordinates, we solve $\mis$ in $O(n^2)$ time and show this to be the  best possible under the Strong Exponential Time Hypothesis. For the intersection graph of $n$ $\el$-shapes in the plane, we give a $(4\cdot \log \opt)$-approximation algorithm for $\mis$ (where $\opt$ denotes the size of an optimal solution), improving the previously best-known $(4\cdot \log n)$-approximation algorithm of Biedl and Derka (WADS 2017).
\end{abstract}


\end{frontmatter}

\section{Introduction}
\label{sec:introduction}
Let $G=(V,E)$ be an undirected graph with \rev{$|V |=n$}. The graph $G$ is \emph{weighted} if each \rev{vertex in $V$} 
is associated with a non-negative value, called its \emph{weight}. A set \rev{$S\subseteq V$} is an \emph{independent set} if no two vertices in $S$ are adjacent. The objective of the \emph{Maximum Independent Set} ($\mis$) problem is to compute a maximum-cardinality independent set of $G$. \rev{The \emph{weighted $\mis$} problem seeks for an independent set $S$ that maximizes the sum of the weights of the vertices in $S$}.  The $\mis$ problem is $\np$-\rev{hard} and it is known that no approximation algorithm with approximation factor within $|V |^{1-\epsilon}$ is possible for any $\epsilon>0$, unless $\mathsf{P}=\np$~\cite{DBLP:conf/focs/Hastad96}. The inapproximability of the $\mis$ problem has motivated a rich body of research to study the $\mis$ problem on the intersection \rev{graphs} of geometric objects. 

Let $O$ be a set of $n$ geometric objects in the plane. Then the \emph{intersection graph} of $O$ has the objects in $O$ as its vertices and two vertices $o_i,o_j\in O$ are adjacent in the graph if and only if $o_i\cap o_j\ne \emptyset$. If $O$ is a set of curves in the plane (resp., a set of chords of a circle), then the intersection graph of $O$ is called a \emph{string graph} (resp., \emph{circle graph}); see Figure~\ref{fig:graphClasses}(b--c) for an example. 

\subsection{Background}

Here we briefly review the related research on computing $\mis$ for  various  classes of intersection graphs, which motivated our work. 

The $\mis$ problem is NP-hard for string graphs, even when the strings are straight line segments~\cite{KratochvilN90}. For string graphs with any two strings intersecting at most a constant number of times, Fox and Pach~\cite{DBLP:conf/soda/FoxP11} gave an approximation algorithm with an approximation factor of $n^\epsilon$. 
 The $\mis$ problem has been studied on intersection \rev{graphs} of other geometric objects such as disks and squares~\cite{DBLP:journals/siamcomp/ErlebachJS05}, rectangles~\cite{DBLP:conf/soda/ChalermsookC09},  general convex objects~\cite{DBLP:journals/comgeo/AgarwalM06}, and pseudo-disks~\cite{DBLP:journals/dcg/ChanH12}.

In this paper we   examine the $\mis$ problem on the class of outerstring graphs, which is  defined as follows.

\begin{definition}[Outerstring Graph~\cite{DBLP:journals/jct/Kratochvil91}]
Graph $G$ is called an \emph{outerstring graph} if it is \emph{an}  intersection graph of a set of curves that lie inside a disk such that each curve intersects the boundary of the disk in one of its endpoints.
\end{definition}

Figure~\ref{fig:graphClasses}(d) shows an example of an outerstring graph. A string representation of a graph is called \emph{grounded}, if one endpoint of each string is attached to a \emph{grounding line} $\ell$ and all strings lie on one side of $\ell$. For example, a graph $G$ is called a \emph{grounded segment graph}, if it is the intersection graph of a set of segments such that each segment is attached to a grounding line $\ell$ at one of its endpoints and all segments lie on one side of $\ell$; see Figure~\ref{fig:graphClasses}(e).

\begin{figure}[t]
\centering
\includegraphics[width=\textwidth]{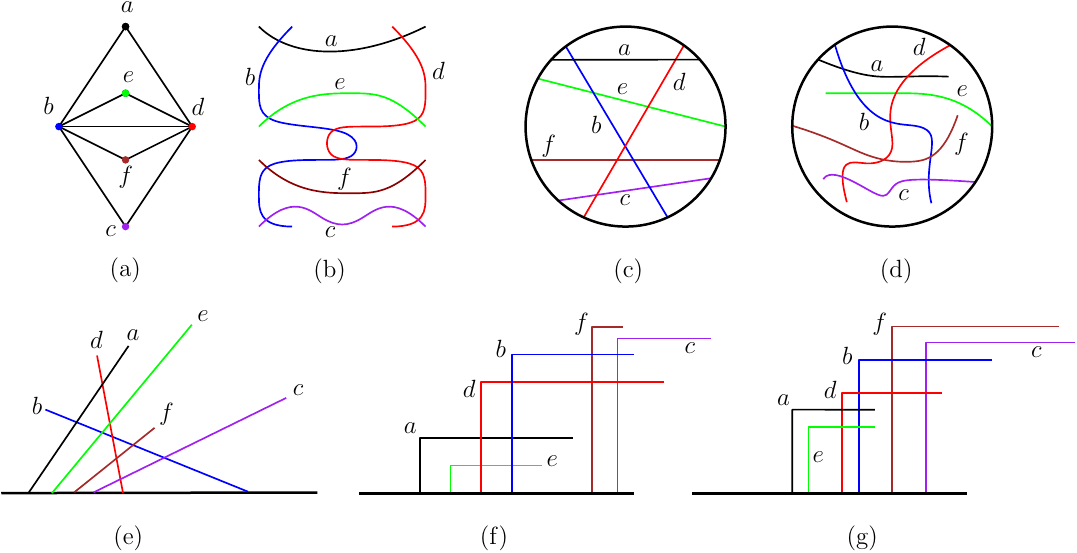}
\caption{(a) A graph $G$ with six vertices. (b) A string graph, (c) a circle graph, (d) an outerstring graph, (e) a \rev{grounded} segment graph, (f) a grounded $\el$, and (g) a grounded square-$\el$ representation  of $G$.}
\label{fig:graphClasses}
\end{figure}

We will relate the time-complexity of computing $\mis$ in outerstring graphs to the time-complexity of computing $\mis$ \rev{in} circle graph representations. Gavril~\cite{DBLP:journals/ipl/Gavril73} presented an $O(n^3)$\rev{-time} algorithm for solving the $\mis$ problem on circle graphs. Subsequent improvement reduced the \rev{time} complexity to $O(n^2)$~\cite{Supowit87,AsanoIM91}. Several algorithms exist with running time sensitive to various graph  parameters, e.g.,   $O(nd)$ time~\cite{DBLP:journals/dam/ApostolicoAH92,DBLP:conf/isaac/Valiente03}, or $O(n \min\{d, \alpha\})$ time~\cite{DBLP:journals/ipl/NashG10}. Here $d$ is a parameter known as the density of the circle graph, and $\alpha$ is the independence number of the circle graph. However, no truly subquadratic-time algorithm (i.e., an $O(n^{2-\delta})$-time algorithm where $\delta>0$) is known for the $\mis$ problem on circle graphs.

Although recognizing an intersection graph may require $\Theta(n^2)$ time (since there could be $\Theta(n^2)$ edges), the $\mis$ problem can be solved faster if an intersection representation is given. For example, $\mis$ in an interval graph representation can be \rev{computed} in $O(n)$ time~\cite{Frank75}. Although recognizing grounded segment graphs is $\exists\mathbb{R}$-complete~\cite{DBLP:journals/jgaa/CardinalFMTV18},  given an outerstring representation \rev{with strings as polygonal chains of straight line segments}, one can solve the weighted $\mis$ problem in $O(s^3)$ time, where $s$ is the number of segments in the representation~\cite{DBLP:journals/comgeo/KeilMPV17}. For grounded segment graphs, this yields a time complexity of $O(n^3)$, where $n$ is the number of vertices in the grounded segment graph. Although the strings in a grounded segment graph are straight line segments, no faster algorithm is known for this case. Thus a natural question is to ask whether one can prove non-trivial lower bounds on the time complexity of the $\mis$ problem for outerstring graphs or simpler variants of such graphs.

An $\el$-shape is the union of a vertical segment and a horizontal segment that share an endpoint; hence, there are four possible types of $\el$-shapes: $\{\ul,\ur,\llc,\lr\}$. A graph is called a \emph{B$_1$-VPG graph} if it is the intersection graph of a set of $\el$-shapes in the plane. This class of string graphs belongs to a larger class called the Vertex intersection of Paths on a Grid (VPG) and denoted by B$_k$-VPG, where $k$ indicates the maximum number of bends each path can have in the grid representation~\cite{DBLP:journals/jgaa/AsinowskiCGLLS12}. These graphs and their relatives have been studied extensively in terms of recognition problems (e.g., see~\cite{DBLP:conf/soda/GoncalvesIP18,DBLP:journals/dam/FelsnerKMU16,DBLP:conf/wg/ChaplickJKV12,DBLP:journals/jgaa/AsinowskiCGLLS12}). Recently,  there has been an increasing attention on studying optimization problems on these graphs; see~\cite{DBLP:conf/mfcs/BandyapadhyayM018,DBLP:conf/waoa/Mehrabi17,DBLP:conf/wads/BiedlD17a,DBLP:journals/corr/Mehrabi17} and the references therein. For the $\mis$ problem, it is known that the problem is $\np$-\rev{hard} on B$_k$-VPG graphs even when $k=1$ \rev{ and a representation is given}~\cite{DBLP:conf/cocoa/LahiriMS15}. The previously best-known approximation  algorithms \rev{(for the $\mis$ problem on B$_1$-VPG graphs)} have factor $4\cdot\log n$~\cite{DBLP:conf/wads/BiedlD17a,DBLP:journals/corr/Mehrabi17}. Combining B$_1$-VPG and grounded string graphs, we consider the $\mis$ problem on \emph{grounded $\el$} and \emph{grounded square-$\el$} graphs.

\begin{definition}[Grounded $\el$ and Grounded Square-$\el$ Graphs.]
 \rev{A graph} $G$ is called a \emph{grounded $\el$ graph} if $G$ is the intersection \rev{graph} of a set of $\el$-shapes such that each $\el$-shape is of type $\ul$ and the lower endpoint of the vertical segment of each $\el$-shape is attached to a grounding line $\ell$. If the vertical and horizontal segments of every $\el$-shape in a grounded $\el$ representation of $G$ have the same length, then we call $G$ a \emph{grounded square-$\el$ graph}.
\end{definition}

See Figure~\ref{fig:graphClasses}(f--g) for examples of these graphs. Finally, for the $\mis$ problem on a set of $n$ rectangles, Chalermsook and Chuzhoy~\cite{DBLP:conf/soda/ChalermsookC09} gave \rev{a}  $(\log \log n)$-approximation algorithm for the unweighted version of the problem. For the weighted version of the problem, the best approximation factor is $O(\log n/\log \log n)$ due to Chan and Har-Peled~\cite{DBLP:journals/dcg/ChanH12}.  
\rev{A collection of geometric objects in the plane are called \emph{pseudo-disks} if  for every pair of these objects, the  boundaries intersect at most twice. 
For pseudo-disks in the plane, there exists a  PTAS for the unweighted case and a
constant-factor approximation for the weighted case~\cite{DBLP:journals/dcg/ChanH12}. Thus for restricted cases, where a set of $\el$-shapes in the plane can be extended to form a collection of pseudo-disks that represents the same intersection graph, one can leverage these algorithms to compute approximate solutions.}

\subsection{Contributions} We now summarize our contribution in $C_1$--$C_3$.

\paragraph{C1 (Section~\ref{sec:gsr})} We first examine the time-complexity of the $\mis$ problem on the grounded segment graphs with respect to its relation to the $\mis$ problem in circle graphs. Middendorf and Pfeiffer~\cite{DBLP:journals/dm/MiddendorfP93} showed that every intersection graph of $\el$-shapes of types $\ul$ and $\llc$ (not necessarily grounded) can be transformed into  a segment representation.  If the $\el$-shapes  are grounded, then the transformation yields a grounded segment graph. Since every circle graph is a grounded $\el$ graph~\cite{DBLP:journals/corr/abs-1808-04148}, they are also grounded segment graphs. However, the transformation~\cite{DBLP:journals/dm/MiddendorfP93} into the grounded segment representation is by an inductive proof, and it is unclear whether the constructed  representation can be encoded in a subquadratic number of bits. We show that the $\mis$ problem in a circle graph representation is $O(n\log n)$-time reducible to the $\mis$ problem in an implicit representation of a grounded segment graph, where the representation takes $O(n\log n)$ bits. This indicates that solving $\mis$ in such grounded segment representations is as hard as solving $\mis$ in circle graph representations.

\paragraph{C2 (Sections~\ref{sec:gSquare}--\ref{sec:dp})} Since grounded $\el$ graphs include  circle graphs, we examined a simpler variant: grounded square-$\el$ graphs. We show that there exist grounded square-$\el$ graphs (resp., grounded $\el$ graphs) that are not circle graphs (resp., grounded square-$\el$ graphs). Although grounded square-$\el$ is a  simpler variant, we prove that it includes the circle graphs. In fact, we give an $O(n\log n)$-time reduction, showing that $\mis$ in grounded square-$\el$ representations is at least as hard as $\mis$ in circle graph representations \rev{under the  assumption that finding an $\mis$ in circle graph requires $\Omega(n\log n)$ time}. In contrast, for the grounded string representations where the strings are $y$-monotone simple polygonal paths of constant length with segments at integral coordinates, we can solve $\mis$ in $O(n^2)$ time. Assuming the Strong Exponential Time Hypothesis (SETH)~\cite{DBLP:journals/jcss/ImpagliazzoPZ01}, we show that 
 \rev{ the $\mis$ problem on an outerstring representation with $s$ strings cannot be solved in $O(s^{2-\delta})$ time, even when each string has one bend.}


\paragraph{C3 (Section~\ref{sec:approxAlg})} We give a $(4\cdot \max\{1,\log \opt\})$-approximation algorithm for the weighted $\mis$ problem on the intersection graph of a set of $n$ $\el$-shapes in the plane. This improves the previously best-known algorithm, which has an approximation factor of $4\cdot\log n$~\cite{DBLP:conf/wads/BiedlD17a,DBLP:journals/corr/Mehrabi17}. Moreover, our algorithm can be used to obtain a simple $(4\cdot \max\{1,\log \opt\})$-approximation algorithm for the weighted $\mis$ problem on a set of $n$ axis-parallel rectangles in the plane.

\section{$\mis$ on Grounded Segment Representations}
\label{sec:gsr}
In this section, we show that the $\mis$ problem in a circle graph representation is $O(n\log n)$-time reducible to the $\mis$ problem in a representation of a grounded segment graph, where the representation takes $O(n\log n)$ bits. This indicates that solving $\mis$ on grounded segment   representations could be as hard as solving $\mis$ on circle graph representations. 

An \emph{overlap graph} is an intersection graph of intervals, where two vertices are adjacent if and only if their corresponding intervals properly \rev{intersect} (i.e., the intersection is non-empty but 
neither \rev{interval} contains the other). Gavril~\cite{DBLP:journals/ipl/Gavril73} showed that a graph is a circle graph if and only if it is an overlap graph. Given the circle graph representation, one can find an overlap representation in linear time by computing the shadow of each chord on a horizontal line below the circle, assuming the point light source is at the apex of the circle as illustrated in Figure~\ref{fig:circletograoundedseg}(a--c). \rev{We now show that} 
 the overlap representation can be transformed into a grounded segment representation in \rev{$O(n \log n)$} time.

\begin{figure}[pt]
\centering
\includegraphics[width=.75\textwidth]{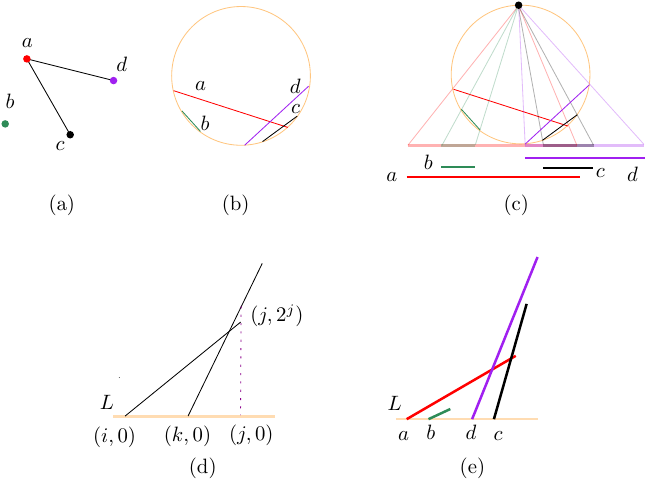}
\caption{(a) A circle graph $G$. (b) A circle graph representation of $G$. (c) Transformation into an overlap graph. (d)-(e) Transformation into a grounded segment graph\rev{, where $L$ is the ground line}. We only show a schematic representation for space constraints.}
\label{fig:circletograoundedseg}
\end{figure}

We assume that the circle graph representation is non-degenerate, i.e., no two chords share a common endpoint. Consequently, the overlap representation is also non-degenerate. We now sort the endpoints of the intervals and relabel them with integral coordinates. For each interval $[i,j]$ in the overlap graph, we define a line segment with coordinates $(i,0),(j,2^j)$. Note that all the segments are grounded at the line $y=0$; i.e., line $L$ in Figure~\ref{fig:circletograoundedseg}(d). Moreover, it is straightforward to encode the representation implicitly in $O(n\log n)$ bits (note that an explicit representation would require $O(n^2)$ bits). Let the resulting representation be $\mathcal{R}$. In the proof of the following theorem we show that $\mathcal{R}$ is the required grounded segment representation.

\begin{theorem}
\label{thm:seg}
Given a circle graph representation with $n$ chords, in $O(n\log n)$ time one can transform it into an implicit grounded segment representation, which uses $O(n\log n)$ bits. Thus, the $\mis$ problem on grounded segment representations is at least as hard as the $\mis$ problem on circle graph representations.
\end{theorem}
\begin{proof}
Consider the representation $\mathcal{R}$ constructed from the overlap representation of the circle graph. It is straightforward to observe that if two intervals do not intersect in the overlap graph, then the corresponding segments do not intersect in $\mathcal{R}$. We now need to prove that if two intervals properly intersect, then the corresponding segments intersect in $\mathcal{R}$; otherwise, one interval contains the other and the segments do not intersect in $\mathcal{R}$.

Let $[i,j]$ and $[k,\ell]$ be two intervals that properly intersect; i.e., $i<k<j<\ell$, and let $s_{[i,j]}$ and $s_{[k,\ell]}$ be the corresponding segments. Note that $s_{[i,j]}$ intersects the line $x=j$ at  height $2^j$. Hence, $s_{[k,\ell]}$ will intersect $s_{[i,j]}$ if it intersects the line $x=j$ at the same or  a higher point. Therefore, we need to show that  $\frac{2^\rev{\ell}}{\rev{\ell}-k} (j-k) \ge 2^j$ holds. Observe that
\[
\frac{2^\rev{\ell}}{\rev{\ell}-k} (j-k) \ge 2^j
  \Leftrightarrow (j-k) - \frac{(\rev{\ell}-k)}{2^{\rev{\ell}-j}} \ge 0  
  \Leftrightarrow 1-\left(\frac{(\rev{\ell}-j)}{(j-k)2^{\rev{\ell}-j}} +\frac{1}{2^{\rev{\ell}-j}}\right)   \ge 0.
\]
Since $(\ell-j)\ge 1$ and $(j-k)\ge 1$, the above condition will hold for any integral $j,k,\ell$, and hence the segments will intersect.

Finally, if the interval $[i,j]$ contains the interval $[k,\ell]$, i.e., $i<k<\ell<j$, then the height of 
$s_{[k,\ell]}$ at $x=\ell$ is $2^\ell$, whereas the height of $s_{[i,j]}$ is $\frac{2^j}{j}\ell = 2^\ell\cdot 2^{j-\ell}\left(\frac{\ell}{j}\right) = 2^\ell\left(\frac{\ell2^{j}}{j2^{\ell}}\right)$. Since $j>\ell$, for any integral   $j,\ell$, the height of $s_{[i,j]}$ at $x=\ell$ will be larger than that of $s_{[k,\ell]}$. Hence, the segments will not intersect.~$\Box$
\end{proof}

\section{$\mis$ on Grounded Square-$\el$ Representations}
\label{sec:gSquare}
In this section, we show that solving $\mis$ in a circle graph representation is $O(n\log n)$-time reducible to solving $\mis$ in a grounded square-$\el$ representation.

Given a circle graph representation, we first compute the corresponding overlap graph in the same way as we did in Section~\ref{sec:gsr}, and relabel the endpoints with integral coordinates from $0$ to $2n-1$. Since we assume the circle graph representation to be non-degenerate, no two endpoints in the overlap graph share the same $x$-coordinate. 

We now transform this into a  grounded square-$\el$ representation. The idea is to process  the intervals in the order of their endpoints, and sometimes  shifting the endpoints by a certain offset $\gamma$ to avoid unnecessary crossings.  We now give \rev{a} formal description of the steps of the construction by $S_1$--$S_3$.\smallskip


\textbf{Step $S_1$.} Initialize an empty list $Q$, and then process the intervals in the increasing order of the $x$-coordinates of their left endpoints. \rev{At the end of the processing, $Q$ will contain a set of tuples that correspond to a set of disjoint intervals. Each tuple is of the form $(X,\gamma)$, where $X$ is an interval and $\gamma$ is an offset to be used later to create intermediate space between pairs of intervals. We now describe how to process an interval.} 
    
While processing an interval $I = [I_\ell,I_r]$, we first find the closest non-intersecting interval $J = [J_\ell,J_r]$ to the left of $I$. If no such interval exists, then we continue processing the next interval. Otherwise, let $(X,\gamma)$ be the tuple at the end of the list $Q$ (assume a dummy tuple $(\Phi, 0)$ if the list is empty). If $J\not =X$, then append a new tuple $(J, J_\ell+\gamma$) to $Q$. 

\textbf{Step $S_2$.} For each pair of consecutive tuples $(A,\alpha)$ and $(B,\beta)$ in $Q$, update the $x$-coordinates of the endpoints originally lying in $[A_r+1,\rev{B_r}]$ by adding the integer $\alpha$. \rev{Later, we will use these updated $x$-coordinates to construct the $\el$-shapes. A spacing of $\alpha$ would ensure that the $\el$-shapes corresponding to $A$ and $B$ will remain disjoint.} Finally, for the last tuple $(X,\gamma)$,  update the $x$-coordinates of the endpoints originally lying in $[X_r+1,+\infty]$, by adding the integer $\gamma$. Figure~\ref{fig:overlaptoL}(a--b) illustrate  this step. 

\textbf{Step $S_3$.} For each interval $[I_\ell, I_r]$ in the increasing order of their left endpoints, construct a square-$\el$ shape with endpoints $(\frac{I_\ell}{2},-\frac{I_\ell}{2})$ and  $(I_r+\frac{I_\ell}{2}, -\frac{I_\ell}{2})$, and create the bend point at $(I_\ell+\frac{I_r-I_\ell}{2}, \frac{I_r-I_\ell}{2})$. See Figure~\ref{fig:overlaptoL}(c).\smallskip


By $S_3$, it is straightforward to see that all the shapes are grounded on the line $x+y=0$. Let $\Gamma$ be the resulting grounded  square-$\el$ representation. The following lemma claims the correctness of the representation. 

\begin{figure}[t]
\centering
\includegraphics[width=\textwidth]{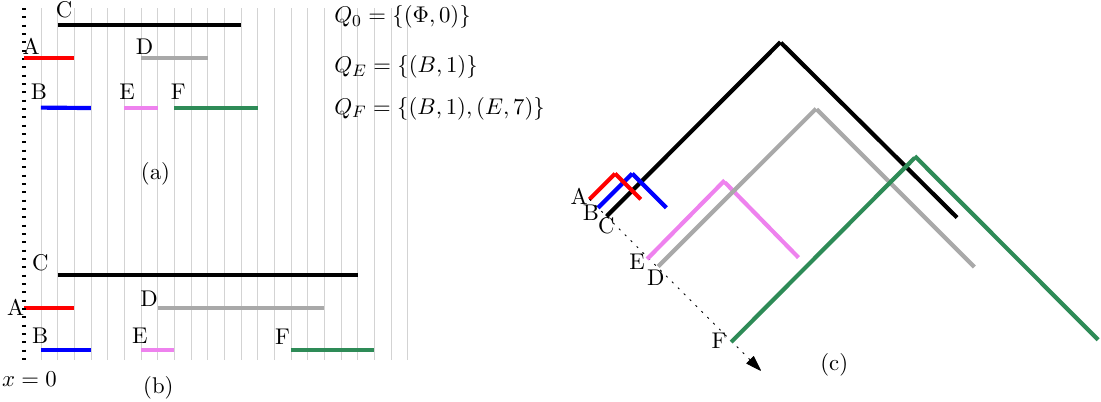}
\caption{(a) An overlap representation. $Q_0$ is the initial empty list. $Q_E$ and $Q_F$ are the lists after processing the intervals $E$ and $F$, respectively. After step $S_1$, the list $Q$ is $\{(B,1),(E,\rev{7})\}$.     (b)  Modification after step $S_2$. All the endpoints in $[B_r+1,E_l]$ have been shifted 1 unit, and all the endpoints in $[E_r+1,+\infty]$ have been shifted \rev{7} units. (c) The grounded square-$\el$ representation constructed at $S_3$; $A$ is grounded at $(0,0)$.}
\label{fig:overlaptoL}
\end{figure}

\begin{lemma}
\label{lem:helper}
The graph represented by $\Gamma$ is the same as the graph represented by the overlap representation.
\end{lemma}
\begin{proof}
Let $G$ be the graph corresponding to the input overlap representation. \rev{Let $H_k$ be the  subgraph of $G$ induced by the first $k$ intervals processed in $S_3$. Let $\Gamma_k$ be the drawing obtained after processing these intervals. It now suffices to prove that $\Gamma_k$ is a grounded square-$\el$ representation of $H_k$.}


\rev{In the base case when $k=1$ is trivial, i.e., the graph $H_1$ consists of a single interval and in $S_3$, we construct a single square-$\el$}. 
\rev{We now} assume that the claim holds for $H_1,\ldots,H_{k-1}$, where $k>1$. \rev{Let $B$ be the $k$th interval and $A$ be an interval processed before $B$. Let $a$ and $b$ be the corresponding vertices in $H_k$,  and 
let } $A',B'$ be the modified intervals (computed in $S_2$). For any interval $I$, let  $L(I)$ be the square-$\el$ shape constructed as in $S_3$. We now consider the following cases.

\textbf{Case 1 ($a$ and $b$ are adjacent in $H_k$):} In this case $A$ and $B$ properly intersect; i.e., neither contains the other. Let $J$ be the closest non-overlapping interval to the left of $B$. \rev{In $S_2$,} any interval having an endpoint between $[J_r+1,B_\ell]$ will be shifted together with $B$. Thus the  intersections of these \emph{modified}  intervals  with \rev{$B'$} remain valid. All the other intervals who were intersecting $B$ but did not have an endpoint in $[J_r+1,B_\ell]$ will also have their right endpoint shifted \rev{in $S_2$}. Hence $A'$ must properly intersect $B'$. We now show that the shift in $S_2$ keeps the ordering of the intervals in $H_{k-1}$ intact.

Consider a pair of vertices in $p,q$ in $H_{k-1}$, and let $P$ and $Q$ be their corresponding intervals. Let $P'$ and $Q'$ be the modified intervals in $S_2$. 

If $p$ and $q$ are adjacent, then $P$ and $Q$ must properly intersect. The endpoints of $P$ and $Q$ can only get extended to the right, and if so, then they must shift by the same amount. Hence, $P'$ and $Q'$ must maintain the endpoint orders after the shift, and $L(P')$ and $L(Q')$ must intersect. 

Now consider the case when $p$ and $q$ are not adjacent. If  one of $P$ and $Q$ contains the other, then the same argument holds. If  neither contains the other, then the offset may only increase their distance. Therefore, if $L(P)$ and $L(Q)$ do not intersect, then  $L(P')$ and $L(Q')$ cannot intersect.

\textbf{Case 2 ($a$ and $b$ are non-adjacent in $H_k$):} In this case either $A$ and $B$ do not intersect, or $A$ contains $B$ (note that $B$ cannot contain $A$).  

\rev{First assume that} $A$ contains $B$. Since the shift in $S_2$ does not change the relative order of the interval endpoints, $A'$ must contain $B'$. \rev{From the coordinate computations of the $\el$-shapes in $S_3$, it is straightforward to observe that $L(A')$ and $L(B')$ do not intersect.} 


Assume now that $A$ and $B$ do not intersect. Recall  that $B$ has been processed after $A$. While we processed $B$  in $S_1$, we first computed the closest interval $J$ to the left of $B$. Hence $A_r\le J_r$. In $S_2$, we ensured that the endpoints of $B$ are shifted to the right by at least an amount of $J_\ell+\gamma$. Here, $\gamma$ corresponds to the overall shift   to accommodate the segments that were processed before $J$, and the term $J_\ell$ is to avoid the crossing between $L(J')$ and $L(B')$. Since $A_r\le J_r$, $L(A')$ and $L(B')$ cannot intersect. Using the argument of Case 1, observe that such shifting still maintains a valid representation for $H_{k-1}$.~$\Box$
\end{proof}

\begin{theorem}
Given a circle graph representation with $n$ chords, in $O(n\log n)$ time one can transform it into a grounded square-$\el$ representation. Thus, the $\mis$ problem on grounded square-$\el$ representations is at least as hard as the $\mis$ problem on circle graph representations. 
\end{theorem}
\begin{proof}
By Lemma~\ref{lem:helper}, one can construct the required grounded square-$\el$ representation by following $S_1$--$S_3$. We compute two sorted arrays, one for the left endpoints and the other for the right endpoints of the intervals in the overlap representation. The sorting takes $O(n\log n)$ time.  We use these arrays to \rev{find the closest non-intersecting interval} in step  $S_1$ in $O(\log n)$ time by performing a binary search. We need only $O(n)$ \rev{such searches}, and hence $O(n\log n)$ time in total. Steps $S_2$--$S_3$ \rev{iterate over all the intervals, and hence} take $O(n)$ time. Therefore, the running time of the  overall transformation is bounded by $O(n\log n)$.~$\Box$
\end{proof}

Our reduction shows that every circle graph is a grounded square-$\el$ graph. However, the reverse is not true. Even, there are grounded $\el$ graphs that are not grounded square-$\el$ graphs.
\begin{theorem}
\label{thm:notAllAre}
There are grounded square-$\el$ graphs that are not circle graphs. Moreover, there are grounded $\el$ graphs that are not grounded square-$\el$ graphs.
\end{theorem}

\begin{figure}[t]
\centering
\includegraphics[width=\textwidth]{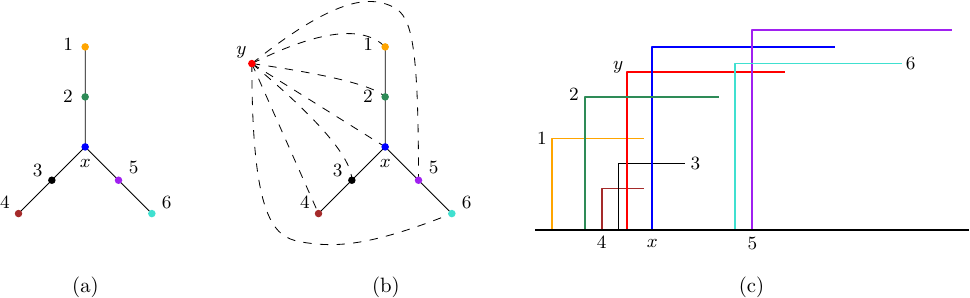}
\caption{(a) A graph $H$, (b) the graph $H^+$ obtained from $H$, and (c) a grounded square-$\el$ representation of $H^+$.} 
\label{fig:notCircle}
\end{figure}

\begin{proof}
We first show that not all grounded square-$\el$ graphs are circle graphs. For a graph $G$, let $G^+$ denote the graph obtained from $G$ by adding a new vertex $y$ to the graph and connecting it to every vertex of $G$; it is known that $G$ is a permutation graph if and only if $G^+$ is a circle graph~\cite{GolumbicBook}. Now, consider the graph $H$ shown in Figure~\ref{fig:notCircle}(a). Limouzy~\cite{DBLP:conf/isaac/Limouzy10} proved that $H$ is not a permutation graph. Consequently, the graph $H^+$ (shown in Figure~\ref{fig:notCircle}(b)) is not a circle graph. However, a grounded square-$\el$ representation of $H^+$ is shown in Figure~\ref{fig:notCircle}(c).

\begin{figure}[h]
\centering
\includegraphics[width=.8\textwidth]{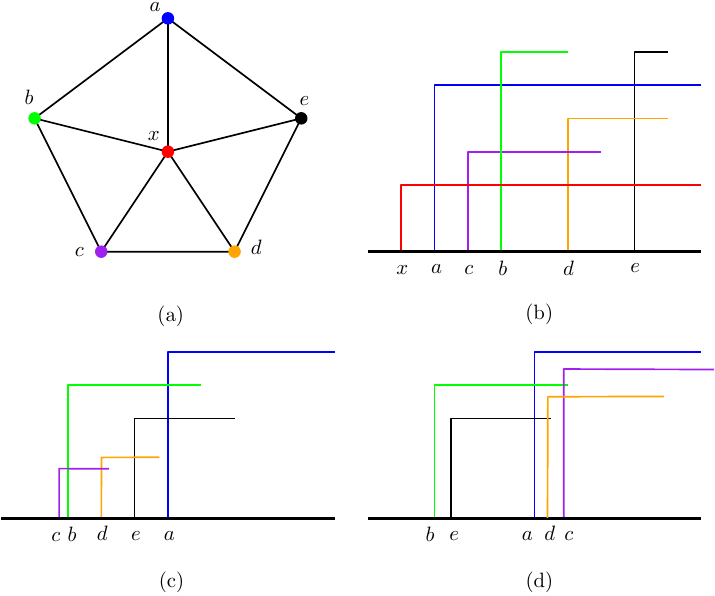}
\caption{An illustration for the proof of Theorem~\ref{thm:notAllAre}.} 
\label{fig:w5}
\end{figure}

We now show that there are grounded $\el$ graphs that are not grounded square-$\el$ graph. To this end, we show that $W_5$ (i.e., the wheel graph of order 5 as shown in Figure~\ref{fig:w5}(a)) is a grounded $\el$ graph, but not a grounded square-$\el$ graph. 

A grounded $\el$ graph representation of $W_5$ is shown in Figure~\ref{fig:w5}(b). Suppose now for a contradiction that $W_5$ shown in Figure~\ref{fig:w5}(a) has a grounded square-$\el$ graph representation. \rev{The idea is to first explore possible square-$\el$ representations of the 5-cycle on the outerface of $W_5$, and then show that one cannot add  the $\el$-shape corresponding to the central vertex $x$ in these representations.}

Consider the set of $\el$-shapes induced by the 5-cycle on the outerface and assume w.l.o.g. that $a$ is the highest $\el$-shape. Since $a$ is the highest $\el$-shape, both of its adjacent $\el$-shapes $b$ and $e$ must intersect $a$ from the left in such a way that $b$ and $e$ do not intersect each other. Assuming w.l.o.g. that $b$ is to the left of $e$.  
 \rev{We now distinguish two cases considering whether $c$ intersects the horizontal or the vertical segment of $b$.} 

\rev{\textbf{Case 1 ($c$ intersects the vertical  segment of $b$):} 
 In this case,}
 $c$ must be to the left of $b$, as shown in Figure~\ref{fig:w5}(c). Then  one can check that the only possibility for the $\el$-shape $d$ is to be between $b$ and $e$. \rev{This gives us a square-$\el$ representation of the outerface of $W_5$. We now show that the $\el$-shape $x$ cannot be added in this representation to intersect all other $\el$-shapes.} 
 
\rev{If $x$ is to the left of $d$, then to intersect $d$, $x$ must have a height smaller than that of $d$. Consequently, $x$ cannot intersect $a$. Therefore, $x$ must be to the right of $d$. If the height of $x$ is smaller than that of $e$, then $x$ cannot intersect $b$. If the height of $x$ is larger than that of $e$, then $x$ cannot intersect both $c$ and $e$ at the same time. Hence we cannot add $x$ to complete a square-$\el$ representation of $W_5$.}

\rev{\textbf{Case 2 ($c$ intersects the horizontal  segment of $b$):}
Note that $a$ is the highest $\el$-shape and  $c$ must not intersect $a$. Hence $c$ must be to the right of $a$, as shown in Figure~\ref{fig:w5}(d). Since the $\el$-shape $d$  does not intersect $a$ but intersects both $e$ and $c$,  it must be positioned between $a$ and $c$.  This gives us a square-$\el$ representation of the outerface of $W_5$. We now show that the $\el$-shape $x$ cannot be added in this representation to intersect all other $\el$-shapes. 
}

\rev{
If $x$ is to the left of $e$, then to intersect $e$, $x$ must have a height smaller than that of $e$. Consequently, $x$ cannot intersect $c$. Therefore, $x$ must be to the right of $e$ and intersect the horizontal segment of both $b$ and $e$. In this scenario, if $x$ lies to the left of $d$, then  $x$ cannot intersect $d$. If $x$ lies to the right of $d$ then  $x$ cannot intersect both $a$ and $c$ at the same time. Hence we cannot add $x$ to complete a square-$\el$ representation of $W_5$.}~$\Box$
\end{proof}

\section{Representations with Bounded-Length Integral Shapes}
\label{sec:dp}
In this section, we consider string representations where the strings are $y$-monotone \rev{(not necessarily strictly monotone)} polygonal paths, the length of each string is bounded by a constant $\kappa$, and all the bends and endpoints are on integral coordinates. We show that the $\mis$ problem on such representations can be \rev{computed} in $O(n^2)$ time, and this is best possible  under the Strong Exponential Time Hypothesis.

\subsection{An $O(n^2)$-time Algorithm}
Here we describe the $O(n^2)$-time algorithm. For simplicity, we first  examine the case when each string is an $\el$-shape of type $\ul$. Denote by $M_p$ an axis-aligned simple $y$-monotone (not necessarily strictly monotone)  polygonal path that satisfies the following three  constraints: 
\begin{enumerate}
    \item [(a)] $M_p$ starts at point $p$, and ends at a point on the line $y=\kappa$.
    \item [(b)] $M_p$ contains at most $2\kappa$ bends, and
    \item [(c)] the length of each line segment in $M_p$ is bounded by $\kappa$.
\end{enumerate}
   Then the number of such distinct strings can be at most $f(\kappa) \in O(1)$ (since $\kappa$ is a constant). Denote the set of such strings by $\mathcal{M}_p$.

We employ a dynamic programming technique, where we express a subproblem with two points $a,b$ on the grounding line and two monotone paths $M_a$ and $M_b$. Figure~\ref{fig:fixedheight}(a) illustrates a subproblem $\mis(a,b,M_a,M_b)$. The subproblem contains all the $\el$-shapes of the given  representation that are in the region between $M_a$ and $M_b$. The left side of the region is open and the right side is closed, hence the $\el$-shape that starts at $a$ must be excluded. While constructing subproblems, we  will ensure that $a$ and $b$ belong to the set of grounding points on the grounding line. The initial problem can be expressed as $\mis(i,j,M_i,M_j)$, where $i$ is a grounding point of a dummy $\el$-shape $I$ lying to the left of all the $\el$-shapes, and  $j$ is the grounding point of the rightmost $\el$-shape. $M_i$  and $M_j$ are two strings that bound all the $\el$-shapes in between.

Given a problem of the form  $\mis(a,b,M_a,M_b)$, we first find a  grounding point $q$ at the median position among the distinct grounding points between $a$ and $b$, as illustrated in Figure~\ref{fig:fixedheight}(b). Note that  $\el$-shapes can share grounding points, and we only consider the distinct points while considering the median point. If $q$ coincides with $b$, then we have the base case where all the $\el$-shapes start at $b$. We thus return 1 or 0 depending on whether there exists a $\el$-shape in the region  between $M_a$ and $M_b$ (this takes $O(n)$  time). Otherwise, we compute the solution using the following recurrence relation. 
\[
\mis(a,b,M_a,M_b)=\max_{M\in \mathcal{M}_q} \mis(a,q,M_a,M) + \mis(q,b,M,M_b).
\]

\begin{figure}[pt]
\centering
\includegraphics[width=\textwidth]{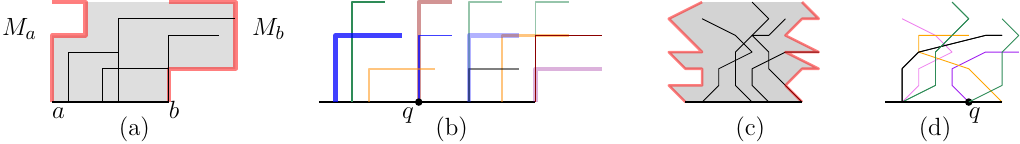}
\caption{Illustration for the dynamic programming. (a) A subproblem. (b) Splitting into subproblems. (c)--(d) General $y$-monotone strings.}
\label{fig:fixedheight}
\end{figure}

To verify the correctness of the recurrence relation, observe that any independent set of $\mis(a,b,M_a,M_b)$ can be partitioned by a string in $M_q$. The size of the dynamic programming table is bounded by $O(n^2)\times O(1)$, where the first  term comes from the choices for $a$ and $b$, and the  $O(1)$ term corresponds to the possible choices for $M_a$ and $M_b$. Computing a base case requires $O(n)$ time.  In the base case, $a$ and $b$ are consecutive on the ground line, and hence there can be at most $O(n)\times f(\kappa)\times f(\kappa)$ distinct base cases, requiring $O(n^2)$ time in total. Computing an entry in the general case requires $f(\kappa)\in O(1)$ time (using constant time table look-up). Hence the running time for the general case is also bounded by $O(n^2)$ in total.

Although we described the algorithm for $\el$-shapes,  it is straightforward to generalize the algorithm for $y$-monotone strings, as illustrated in Figure~\ref{fig:fixedheight}(c)--(d). The only difference is that we need to define $M_p$ as a simple $y$-monotone path. The following theorem summarizes the results of this section.
\begin{theorem}
Let $R$ be a string representation such that the strings are $y$-monotone (not necessarily strict), the length of each string is bounded by a constant, and all the bends and endpoints are on integral coordinates. Then  the $\mis$ problem in $R$ can be solved in $O(n^2)$ time.
\end{theorem}

\subsection{Lower Bound}
The Strong Exponential Time Hypothesis (SETH), introduced by Impagliazzo et al.~\cite{DBLP:journals/jcss/ImpagliazzoPZ01}, has been used to analyze fine-grained time-complexity of problems that lie in P. Under SETH, CNF-SAT on $n$ variables cannot be solved in  $O(2^{n(1-\epsilon )}poly(n))$ time\rev{, where} $\epsilon >0$. The following theorem sates that under SETH, finding $\mis$ in outerstring graphs requires $\Omega(n^{2-\epsilon})$ time.
\begin{theorem}
\label{thm:seth}
Assuming SETH, computing an $\mis$ in an outerstring representation with $n$ strings requires $\Omega(n^{2-\epsilon})$ time, even when each string contains $O(1)$ bends.
\end{theorem}
\begin{proof}
Given an instance \rev{$I$} of CNF-SAT \rev{with $n$ variables and $m$ clauses, we will construct an outerstring representation with $O(2^{n/2} poly(n,m))$ strings such that $I$ admits an affirmative answer if and only if the constructed outerstring representation contains an independent set of size $m$. The construction of the outerstring representation would take $O(2^{n/2} poly(n,m))$ time. Note that the size of the outerstring representation is $N = O(2^{n/2} poly(n,m))$. Assume now that  there exists an $N^{2-\epsilon}$ algorithm, where $\epsilon>0$, for the MIS problem in an outerstring representation. Then we can find a satisfying assignment for $I$ in $N^{2-\epsilon} = O(2^{(2-\epsilon)n/2} poly(n,m)) = O(2^{(1-\epsilon')n} poly(n,m))$ time, where $\epsilon' = \epsilon/2$. This is a contradiction to our SETH assumption. }

\textit{Construction of the outerstring representation:} \rev{The idea of the construction is as follows. We first  partition the $n$ variables of $I$ into two sets $A$ and $B$. Note that the variables in $A$ can have $2^{n/2}$ different truth assignments. We then   construct the outerstrings that correspond two these truth assignments on the left halfplane of the $x=0$ line. Finally, we  construct the outerstrings for the set $B$ on the right halfplane of $x=0$ line symmetrically. We now describe the details.} 

 



\begin{figure}[pt]
\centering
\includegraphics[width=\textwidth]{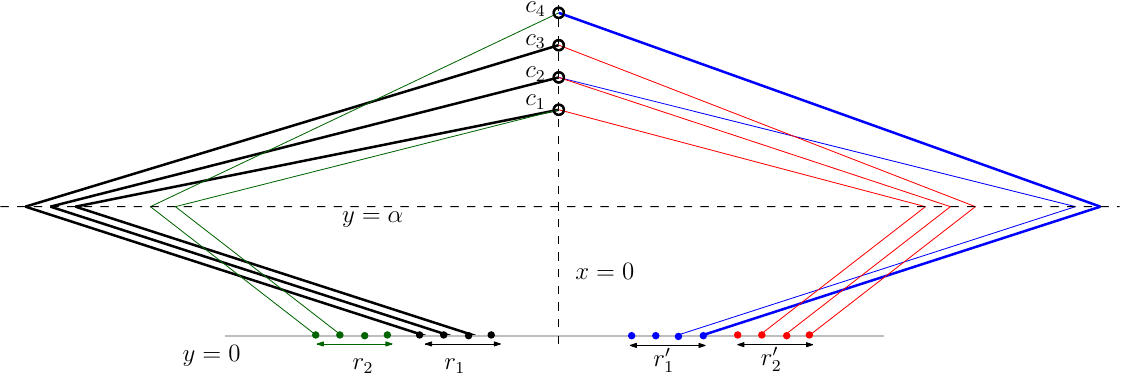}
\caption{Illustration for the proof of Theorem~\ref{thm:seth}\rev{, where $m=4$. The truth assignment $a_1$ satisfies the clauses $c_1,c_2$ and $c_3$. An independent set of size 4 is shown in bold. The corresponding truth assignments $a_1$ and $b_1$ satisfy all the clauses.} } 
\label{fig:seth}
\end{figure}

 Let $c_1,\ldots, c_m$ be the $m$ clauses  \rev{in $I$}, and denote by the `clause-point' $p_i$ the point $(0,2\alpha+i)$, where $\alpha$ is a positive constant.   Let $a_i$, where $1\le i\le 2^{n/2}$, be  the truth value assignments for the variables in  $A$. For $a_1,a_2,\ldots$, assign  intervals $r_1,r_2,\ldots$, each of length $m$,  consecutively on the grounding line, as illustrated in Figure~\ref{fig:seth}. If the assignment corresponding to $a_i$ satisfies a set $S$ of $\beta$ clauses, then we will create $\beta$ outerstrings that start at $r_i$, and each connects to a distinct clause in $I$. It is straightforward to ensure that the strings lie on the left half-plane of $x=0$, and  do not intersect themselves. For each string, we create a bend on the line $y=\alpha$ so that for any $1\le i,j\le 2^{n/2}$, where $i\not =j$, the strings that originate   from $r_i$ intersect those that originate from $r_j$. We construct the strings for the truth assignments $b_1,\ldots,b_{2^{n/2}}$ for the variables of $B$ symmetrically. Let $\mathcal{R}$ be the resulting representation.

\textit{From an independent set to a satisfying assignment:} We now show that an $\mis$ of size $m$ on $\mathcal{R}$   corresponds to an affirmative solution to $I$. Without loss of generality assume that the $\mis$ contains a string $Q$ that starts at some $r_i$. Assume that  $a_i$ satisfies $\beta$ clauses of $I$. 
We take $a_i$ as the  assignment for the variables in $A$.  If $\beta =m$, then we can choose any assignment for the variables in $B$. \rev{Otherwise, $\beta< m$. Since the size of $\mis$ is $m$, there must be $(m-\beta)$ strings on the right-halfplane in the solution that correspond to the remaining set of clauses of $I$. Furthermore, all these $(m-\beta)$ strings correspond to exactly one truth assignment for the variables of $B$.} We thus obtain a  truth assignment for the variables of $B$  that together with $A$, satisfies all the $m$ clauses.

\textit{From a satisfying assignment to an independent set:} If $I$ admits an affirmative answer, then the corresponding assignment of the variables on $A$ and $B$ will correspond to two   intervals $r_a$ and $r_b$ to the left and right half-planes of $x=0$, respectively.  Since these assignments together satisfy all the clauses, choosing all the strings from $r_a$, and all from $r_b$ except those that intersect the ones from $r_a$, will give an independent set of size $m$.~$\Box$
\end{proof}

\section{A $(4\cdot \log \opt)$-Approximation Algorithm}
\label{sec:approxAlg}
In this section, we give a $(4\cdot \max\{1,\log \opt\})$-approximation algorithm for the $\mis$ problem on the intersection graph of a set of $n$ $\el$-shapes. To this end, we first give a $(\max\{1,\log \opt\})$-approximation algorithm for the problem when the input consist of only $\el$-shapes of type $\ul$. We discuss the generalization of our algorithm to the weighted version of the $\mis$ problem and for approximating the $\mis$ problem on rectangles at the end of this section.

Consider the input $\el$-shapes from left to right in the increasing order of the $x$-coordinate of their vertical segment; we denote the $i$th $\el$-shape in this ordering by $\el_i$. For any $1\leq i<j\leq n$, we define $I[i,j]$ as the set of $\el$-shapes $\el_x$ such that (i) $i\leq x\leq j$, and (ii) $\el_x$ does not intersect the line through the vertical segment of $\el_{j+1}$. We add a dummy $\el$-shape $\el_{n+1}$ far to the right such that no input $\el$-shape intersects the line through the vertical segment of $\el_{n+1}$; thus, $I[1,n]$ is the set of all input $\el$-shapes. Moreover, let $\opt[i,j]$ denote the size of an optimal solution for the $\mis$ problem on the set of $\el$-shapes in $I[i,j]$; we denote $\opt[1,n]$ simply by $\opt$. 
Let $I_k$ denote the set of $\el$-shapes that  intersects the line through the vertical segment of $\el_k$. Moreover, let $\opt(I_k)$ be the size of an optimal solution for the $\mis$ problem on the intersection graph induced by the $\el$-shapes in $I_k$.

We define $S[i,j]$ as the \rev{size of the}  solution returned by our algorithm on the $\el$-shapes in $I[i,j]$, for all $1\leq i<j\leq n$. Initially, for every pair $1\leq i<j\leq n$, if $I[i,j]=\emptyset$, then we set $S[i,j]=0$. Then  for every pair $1\leq i<j\leq n$, we check to see if $\opt[i,j]\leq 4$; if so, then we directly store $\opt[i,j]$ in $S[i,j]$. Otherwise, we compute $S[i,j]$ as follows.
\[
S[i,j]=\max\{\max_{i<k< j}S[i,k-1]+S[k+1,j], \opt(I_k)\}.
\]
The algorithm returns $S[1,n]$ as the solution. Computing the actual solution can be done in the standard manner; to this end, we also store the corresponding value of $k$ in $S[i,j]$.

\paragraph{Approximation factor} To show the approximation factor, let $\opts[i,j]$ be the set of $\el$-shapes in $\opt[i,j]$. If $\opt[i,j]\leq 4$, then we have $S[i,j]=\opt[i,j]$. We now prove by induction that for all $1\leq i<j\leq n$, if $\opt[i,j]>4$, then $S[i,j]\geq \opt[i,j]/\log \opt[i,j]$. Suppose that $S[i,j]\geq \opt[i,j]/\log \opt[i,j]$ for all $1\leq i<j\leq n$ for which $4<\opt[i,j]$. 
Let $k^i_j$ be the index such that $\el_{k^i_j}$ is the median of the $\el$-shapes in $\opts[i,j]$ (i.e., each $\opts[i,k^i_j-1]$ and $\opts[k^i_j+1,j]$ contains at most \rev{$\opts[i,j]/2$ } $\el$-shapes). Notice that
\begin{equation}
\label{eq:outerstring}
\opt(I_{k^i_j})\geq |\opts[i,j]\cap I_{k^i_j}|.
\end{equation}
Now, if $\opt[i,k^i_j-1]\leq 4$, then we know that $S[i,k^i_j-1]=\opt[i,k^i_j-1]$. Otherwise, by the induction hypothesis, we have
\begin{equation}
\label{eq:firstInduction}
S[i,k^i_j-1]\geq \frac{\opt[i,k^i_j-1]}{\log \opt[i,j]/2}\geq \frac{|\opts[i,j]\cap I[i,k^i_j-1]|}{\log \opt[i,j]-1}.
\end{equation}
Similarly, if $\opt[k^i_j+1,j]\leq 4$, then we know that $S[k^i_j+1,j]=\opt[k^i_j+1,j]$. Otherwise, by the induction hypothesis, we have
\begin{equation}
\label{eq:secondInduction}
S[k^i_j+1,j]\geq \frac{\opt[k^i_j+1,j]}{\log \opt[i,j]/2}\geq \frac{|\opts[i,j]\cap I[k^i_j+1,j]|}{\log \opt[i,j]-1}.
\end{equation}

Therefore,
\begin{align*}
S[i,j] & = \max\{\max_{i<k<j} S[i,k-1]+S[k+1,j], \opt(I_k)\} \\
    & \geq \max\{S[i,k^i_j-1]+S[k^i_j+1,j], \opt(I_{k^i_j})\} \\
    & \geq \max\{\frac{|\opts[i,j]\cap I[i,k^i_j-1]|+|\opts[i,j]\cap I[k^i_j+1,j]|}{\log \opt[i,j]-1}, |\opts[i,j]\cap I_{k^i_j}|\}\\
    & \geq \max\{\frac{\opt[i,j]-|\opts[i,j]\cap I_{k^i_j}|}{\log \opt[i,j]-1}, |\opt[i,j]\cap I_{k^i_j}|\}.
\end{align*}

The first inequality is because our algorithm tries all values of $i<k<j$, which includes $k^i_j$. Moreover, the second inequality is because of~\eqref{eq:secondInduction}, \eqref{eq:firstInduction} and~\eqref{eq:outerstring}. Now, if $|\opts[i,j]\cap I_{k^i_j}|\geq \opt[i,j]/\log \opt[i,j]$, then we are done. Otherwise,
\begin{align*}
\frac{\opt[i,j]-|\opts[i,j]\cap I_{k^i_j}|}{\log \opt[i,j]-1} & \geq \frac{\opt[i,j]-\opt[i,j]/\log \opt[i,j]}{\log \opt[i,j]-1} \\
    & =\frac{\opt[i,j]}{\log \opt[i,j]}.
\end{align*}
This completes the proof of the induction step. By setting $i=1$ and $j=n$, we have $S[1,n]\geq \opt/\log \opt$. 

\paragraph{Running time} 
\rev{We determine whether $\opt[i,j]\le 4$ for all $1\le i<j\le n$ in $O(n^5)$ time, as follows. We first enumerate all subsets of $\el$-shapes of size 5. Assume that the $\el$-shapes of such a subset form an independent set and let $\el_p$ and $\el_q$ be the leftmost and rightmost $\el$-shapes. Then $\opt[p,q]$ must be larger than $4$ and we mark the pair $(p,q)$. We now   
enumerate all subsets of $\el$-shapes of size at most 4. If such a subset $Q$ is independent and corresponds to a pair $(i,j)$, then we check whether this pair is already marked. If it is marked, then $\opt[i,j] > 4$. Otherwise, we  set  $S[i,j]$ to be  $\max\{S[i,j],|Q|\}$. Hence at the end of this process, for every $i,j$, if $\opt[i,j]\le 4$, then $S[i,j]$ stores $\opt[i,j]$. } 

We now consider computation of $\opt(I_k)$. 
For a fixed triple $i,j$ and $k$, we can compute $\opt(I_k)$ in $O(n^3)$ time because the corresponding graph is an outerstring graph for which $\mis$ can be solved in $O(n^3)$ time~\cite{DBLP:journals/comgeo/KeilMPV17}. Since there are $O(n)$ choices for $k$ for a fixed pair of $i$ and $j$, and $O(n^2)$ entries in the table for $i$ and $j$, the overall running time of the algorithm is $O(n^6)$. We next show how to improve the running time to $O(n^5)$ by performing the following preprocessing. For a fixed triple $i,j$ and $k$, we first compute $\opt(I_k)$ and store the value in a table $T$, and will then do one look-up when computing the corresponding table entry of $S[i,j]$. To this end, we first note that index $j$ is irrelevant for computing $\opt(I_k)$ because for a fixed $i$ and $k$, the set of $\el$-shapes is the same for all $k<j\leq n$. Therefore, for all pairs $1\leq i<k<n$, we compute $\opt(I_k)$ using the algorithm of Keil et al.~\cite{DBLP:journals/comgeo/KeilMPV17} and store it in $T[i,k]$. Since their algorithm takes $O(n^3)$ and there are $O(n^2)$ entries for $T$, the preprocessing step takes $O(n^5)$ overall time. Consequently, this improves the overall running time of computing the entries of table $S$ to $O(n^3)$ and so we have the following lemma.
\begin{lemma}
\label{lem:logOPTApproximation}
There exists an $O(n^5)$-time $(\max\{1,\log \opt\})$-approximation algorithm for the $\mis$ problem on a set of $n$ $\el$-shapes of type $\ul$, where $\opt$ denotes the size of an optimal solution.
\end{lemma}

When the input consists of all four types of $\el$-shapes, we run the algorithm of Lemma~\ref{lem:logOPTApproximation} four times (once for each type of the input $\el$-shapes), and then return the largest solution as the final answer. Clearly, this gives us a $(4\cdot \log \opt)$-approximation algorithm for the original problem and so we have the main result of this section.
\begin{theorem}
\label{thm:logOPTApproximation}
There exists an $O(n^5)$-time $(4\cdot \max\{1,\log \opt\})$-approximation algorithm for the $\mis$ problem on any set of $n$ $\el$-shapes, where $\opt$ denotes the size of an optimal solution.
\end{theorem}

\paragraph{Generalizations} Our algorithm can be generalized in two ways: for the weighted version of the $\mis$ problem on $\el$-shapes, and for the weighted $\mis$ problem on axis-parallel rectangles.
\begin{theorem}
\label{thm:logOPTApproximationGneralization}
There exists an $O(n^5)$-time $(4\cdot \max\{1,\log \opt\})$-approximation algorithm (resp., an $O(n^3)$-time $(\max\{1,\log \opt\})$-approximation algorithm) for the weighted $\mis$ problem on any set of $n$ $\el$-shapes (resp., a set of $n$ axis-parallel rectangles in the plane), where $\opt$ is the size of an optimal solution.
\end{theorem}
\begin{proof}
Suppose that each $\el$-shape has a weight, that is greater than or equal to 1. To apply our algorithm, we now use the ``weighted'' median of the $\el$-shapes in $\opt[i,j]$. Moreover, the algorithm of Keil et al.~\cite{DBLP:journals/comgeo/KeilMPV17} for the $\mis$ problem on outerstring graphs works for weighted outerstring graphs as well. Finally, we can still compute the optimal solution for the weighted $\mis$ problem when $\opt[i,j]\leq 4$ for all \rev{$1\leq i<j\leq n$}. Hence, we have an $O(n^5)$-time $(4\cdot \max\{1,\log \opt\})$-approximation algorithm for the weighted $\mis$ problem.

Next, we show that our algorithm can also be applied to get a $(\log \opt)$-approximation algorithm for the \emph{weighted} $\mis$ on the intersection graph of a set of axis-parallel rectangles in the plane. To see this, we sort the rectangles from left to right by the increasing order of the $x$-coordinate of their left sides, and consider the weighted median. Moreover, we can still compute the optimal solution for the weighted $\mis$ problem when $\opt[i,j]\leq 4$ for all $1\leq i<j\leq j$. To solve the $\mis$ problem on the rectangles in $I_k$, notice that the intersection graph induced by the rectangles in $I_k$ is equivalent to the interval graph obtained by projecting each rectangle of $I_k$ onto the vertical line through the left side of $R_k$, the $k$th rectangle in the ordering. Hence, we can solve the weighted $\mis$ on the rectangles in $I_k$ in $O(n)$ time (given an ordering of these rectangles). The latter improves the overall running time of the algorithm in Theorem~\ref{thm:logOPTApproximation} to $O(n^3)$ because we can now compute all the entries of table $T$ in $O(n^3)$ overall time. Finally, since we have only one type of input rectangles, we do not need to apply our algorithm four times in the case of rectangles and so we have a $(\log \opt)$-approximation algorithm. This completes the proof of the theorem.~$\Box$
\end{proof}

We note that for the case of rectangles, our $(\max\{1,\log \opt\})$-approximation algorithm provides a somewhat simpler algorithm than the one that can be obtained (with the same approximation factor) from the $O(\log n/\log \log n)$-approximation algorithm of Chan and Har-Peled~\cite{DBLP:journals/dcg/ChanH12}.

\section{Conclusion}
\label{sec:conclusion}
In this paper, we studied the fine-grained complexity and approximability of the $\mis$ problem on outerstring graphs and their relatives. Our work gives rise to some natural open questions:
\begin{enumerate}
\item Does there exist a quadratic-time algorithm that can solve the $\mis$ problem on grounded segment or grounded square-$\el$ graphs?
\item Can we improve the approximation factor of the algorithm of Theorem~\ref{thm:logOPTApproximation}?
\item Can we find an $\Omega(n^{2-\epsilon})$-time lower bound under SETH for finding $\mis$ in grounded segment representations?
\end{enumerate}

\rev{
\subsection*{Acknowledgements} We thank the anonymous reviewers for their helpful comments and suggestions, which
improved the presentation of the paper.}

\section*{References}
\bibliography{ref}

\begin{thebibliography}{10}

\bibitem{DBLP:journals/comgeo/AgarwalM06}
Pankaj~K. Agarwal and Nabil~H. Mustafa.
\newblock Independent set of intersection graphs of convex objects in 2d.
\newblock {\em Comput. Geom.}, 34(2):83--95, 2006.

\bibitem{DBLP:journals/dam/ApostolicoAH92}
Alberto Apostolico, Mikhail~J. Atallah, and Susanne~E. Hambrusch.
\newblock New clique and independent set algorithms for circle graphs.
\newblock {\em Discrete Applied Mathematics}, 36(1):1--24, 1992.

\bibitem{AsanoIM91}
T.~Asano, H.~Imai, and A.~Mukaiyama.
\newblock Finding a maximum weight independent set of a circle graph.
\newblock {\em IEICE Transactions}, E74(4):681--683, 1991.

\bibitem{DBLP:journals/jgaa/AsinowskiCGLLS12}
Andrei Asinowski, Elad Cohen, Martin~Charles Golumbic, Vincent Limouzy, Marina
  Lipshteyn, and Michal Stern.
\newblock Vertex intersection graphs of paths on a grid.
\newblock {\em J. Graph Algorithms Appl.}, 16(2):129--150, 2012.

\bibitem{DBLP:conf/mfcs/BandyapadhyayM018}
Sayan Bandyapadhyay, Anil Maheshwari, Saeed Mehrabi, and Subhash Suri.
\newblock Approximating dominating set on intersection graphs of rectangles and
  {L}-frames.
\newblock In {\em proceedings of the 43rd International Symposium on
  Mathematical Foundations of Computer Science (MFCS 2018), Liverpool, {UK}},
  pages 37:1--37:15, 2018.

\bibitem{DBLP:conf/wads/BiedlD17a}
Therese~C. Biedl and Martin Derka.
\newblock Splitting {B}$_2$-{VPG} graphs into outer-string and co-comparability
  graphs.
\newblock In {\em proceedings of the 15th International Symposium on Algorithms
  and Data Structures (WADS 2017), St. John's, NL, Canada}, pages 157--168,
  2017.

\bibitem{BoseCKMMMS19}
Prosenjit Bose, Paz Carmi, J.~Mark Keil, Anil Maheshwari, Saeed Mehrabi,
  Debajyoti Mondal, and Michiel H.~M. Smid.
\newblock Computing maximum independent set on outerstring graphs and their
  relatives.
\newblock In {\em proceedings of the 16th International Symposium on Algorithms
  and Data Structures (WADS 2019), Edmonton, Canada.}, 2019.

\bibitem{DBLP:journals/jgaa/CardinalFMTV18}
Jean Cardinal, Stefan Felsner, Tillmann Miltzow, Casey Tompkins, and Birgit
  Vogtenhuber.
\newblock Intersection graphs of rays and grounded segments.
\newblock {\em J. Graph Algorithms Appl.}, 22(2):273--295, 2018.

\bibitem{DBLP:conf/soda/ChalermsookC09}
Parinya Chalermsook and Julia Chuzhoy.
\newblock Maximum independent set of rectangles.
\newblock In {\em proceedings of the 20th Annual {ACM-SIAM} Symposium on
  Discrete Algorithms (SODA 2009), New York, NY, USA}, pages 892--901, 2009.

\bibitem{DBLP:journals/dcg/ChanH12}
Timothy~M. Chan and Sariel Har{-}Peled.
\newblock Approximation algorithms for maximum independent set of pseudo-disks.
\newblock {\em Discrete {\&} Computational Geometry}, 48(2):373--392, 2012.

\bibitem{DBLP:conf/wg/ChaplickJKV12}
Steven Chaplick, V{\'{\i}}t Jel{\'{\i}}nek, Jan Kratochv{\'{\i}}l, and
  Tom{\'{a}}s Vyskocil.
\newblock Bend-bounded path intersection graphs: Sausages, noodles, and waffles
  on a grill.
\newblock In {\em proceedings of the 38th International Workshop on
  Graph-Theoretic Concepts in Computer Science (WG 2012), Jerusalem, Israel},
  pages 274--285, 2012.

\bibitem{DBLP:journals/siamcomp/ErlebachJS05}
Thomas Erlebach, Klaus Jansen, and Eike Seidel.
\newblock Polynomial-time approximation schemes for geometric intersection
  graphs.
\newblock {\em {SIAM} J. Comput.}, 34(6):1302--1323, 2005.

\bibitem{DBLP:journals/dam/FelsnerKMU16}
Stefan Felsner, Kolja~B. Knauer, George~B. Mertzios, and Torsten Ueckerdt.
\newblock Intersection graphs of {L}-shapes and segments in the plane.
\newblock {\em Discrete Applied Mathematics}, 206:48--55, 2016.

\bibitem{DBLP:conf/soda/FoxP11}
Jacob Fox and J{\'{a}}nos Pach.
\newblock Computing the independence number of intersection graphs.
\newblock In {\em proceedings of the 22nd Annual {ACM-SIAM} Symposium on
  Discrete Algorithms (SODA 2011), San Francisco, CA, USA}, pages 1161--1165,
  2011.

\bibitem{Frank75}
Andr\'{a}s Frank.
\newblock Some polynomial algorithms for certain graphs and hypergraphs.
\newblock In {\em proceedings of the 5th British Combinatorial Conference},
  1975.

\bibitem{DBLP:journals/ipl/Gavril73}
Fanica Gavril.
\newblock Algorithms for a maximum clique and a maximum independent set of a
  circle graph.
\newblock {\em Networks}, 3:261--273, 1973.

\bibitem{GolumbicBook}
Martin~C. Golumbic.
\newblock {\em Algorithmic Graph Theory and Perfect Graphs}, volume vol. 57.
\newblock in: Annals of Discrete Mathematics, North-Holland, Amsterdam, 2nd
  edition edition, 2004.

\bibitem{DBLP:conf/soda/GoncalvesIP18}
Daniel Gon{\c{c}}alves, Lucas Isenmann, and Claire Pennarun.
\newblock Planar graphs as {L}-intersection or {L}-contact graphs.
\newblock In {\em proceedings of the 29th Annual {ACM-SIAM} Symposium on
  Discrete Algorithms (SODA 2018), New Orleans, LA, USA}, pages 172--184, 2018.

\bibitem{DBLP:conf/focs/Hastad96}
Johan H{\aa}stad.
\newblock Clique is hard to approximate within $n^{1-\epsilon}$.
\newblock In {\em proceedings of the 37th Annual Symposium on Foundations of
  Computer Science (FOCS 1996), Burlington, Vermont, USA}, pages 627--636,
  1996.

\bibitem{DBLP:journals/jcss/ImpagliazzoPZ01}
Russell Impagliazzo, Ramamohan Paturi, and Francis Zane.
\newblock Which problems have strongly exponential complexity?
\newblock {\em J. Comput. Syst. Sci.}, 63(4):512--530, 2001.

\bibitem{DBLP:journals/corr/abs-1808-04148}
V{\'{\i}}t Jel{\'{\i}}nek and Martin T{\"{o}}pfer.
\newblock On grounded {L}-graphs and their relatives.
\newblock {\em CoRR}, abs/1808.04148, 2018.

\bibitem{DBLP:journals/comgeo/KeilMPV17}
J.~Mark Keil, Joseph S.~B. Mitchell, Dinabandhu Pradhan, and Martin Vatshelle.
\newblock An algorithm for the maximum weight independent set problem on
  outerstring graphs.
\newblock {\em Comput. Geom.}, 60:19--25, 2017.

\bibitem{DBLP:journals/jct/Kratochvil91}
Jan Kratochv{\'{\i}}l.
\newblock String graphs. {I}. the number of critical nonstring graphs is
  infinite.
\newblock {\em J. Comb. Theory, Ser. {B}}, 52(1):53--66, 1991.

\bibitem{KratochvilN90}
Jan Kratochv{\'{\i}}l and Jaroslav Nesetril.
\newblock Independent set and clique problems in intersection-defined classes
  of graphs.
\newblock {\em Comment. Math. Univ. Carolinae}, 31(1):85--93, 1990.

\bibitem{DBLP:conf/cocoa/LahiriMS15}
Abhiruk Lahiri, Joydeep Mukherjee, and C.~R. Subramanian.
\newblock Maximum independent set on {B}$_1$-{VPG} graphs.
\newblock In {\em proceedings of the 9th International Conference Combinatorial
  Optimization and Applications (COCOA 2015), Houston, TX, USA}, pages
  633--646, 2015.

\bibitem{DBLP:conf/isaac/Limouzy10}
Vincent Limouzy.
\newblock Seidel minor, permutation graphs and combinatorial properties.
\newblock In {\em proceedings of the 21st International Symposium on Algorithms
  and Computation (ISAAC 2010)}, pages 194--205, 2010.

\bibitem{DBLP:conf/waoa/Mehrabi17}
Saeed Mehrabi.
\newblock Approximating domination on intersection graphs of paths on a grid.
\newblock In {\em proceedings of the 15th International Workshop on
  Approximation and Online Algorithms (WAOA 2017), Vienna, Austria}, pages
  76--89, 2017.

\bibitem{DBLP:journals/corr/Mehrabi17}
Saeed Mehrabi.
\newblock Approximation algorithms for independence and domination on
  {B}\({}_{\mbox{1}}\)-{VPG} and {B}\({}_{\mbox{1}}\)-{EPG} graphs.
\newblock {\em CoRR}, abs/1702.05633, 2017.

\bibitem{DBLP:journals/dm/MiddendorfP93}
Matthias Middendorf and Frank Pfeiffer.
\newblock Weakly transitive orientations, {H}asse diagrams and string graphs.
\newblock {\em Discrete Mathematics}, 111(1-3):393--400, 1993.

\bibitem{DBLP:journals/ipl/NashG10}
Nicholas Nash and David Gregg.
\newblock An output sensitive algorithm for computing a maximum independent set
  of a circle graph.
\newblock {\em Inf. Process. Lett.}, 110(16):630--634, 2010.

\bibitem{Supowit87}
K.~J. Supowit.
\newblock Finding a maximum planar subset of a set of nets in a channel.
\newblock {\em IEEE Trans. on CAD of Integrated Circuits and Systems},
  6(1):93--94, 1987.

\bibitem{DBLP:conf/isaac/Valiente03}
Gabriel Valiente.
\newblock A new simple algorithm for the maximum-weight independent set problem
  on circle graphs.
\newblock In {\em proceedings of the 14th International Symposium on Algorithms
  and Computation (ISAAC 2003), Kyoto, Japan}, pages 129--137, 2003.

\end{thebibliography}

\end{document}